\documentclass[10pt,twocolumn,twoside]{IEEEtran}
\usepackage{amsmath}
\usepackage{amsthm}
\usepackage{amsfonts}
\usepackage{amssymb}
\usepackage{bm}
\usepackage{microtype}
\usepackage[margin=0.8in]{geometry}
\usepackage{mathrsfs}
\usepackage{mathtools}
\usepackage{graphicx}
\usepackage{cite}
\usepackage{epstopdf}
\DeclareMathOperator*{\argmax}{arg\,max}

\newtheorem{lemma}{Lemma}
\newtheorem{theorem}{Theorem}

\usepackage{algorithm}
\usepackage{algpseudocode}

\algrenewcommand\alglinenumber[1]{\scriptsize #1}
\def \ETx {\mathcal{E}}
\def \cR {\mathcal{R}}
\def \TRx {\Gamma}	
	
\makeatletter
\let\OldStatex\Statex
\renewcommand{\Statex}[1][3]{%
  \setlength\@tempdima{\algorithmicindent}%
  \OldStatex\hskip\dimexpr#1\@tempdima\relax}
\makeatother
\algdef{SE}[DOWHILE]{Do}{DoWhile}{\algorithmicdo}[1]{\algorithmicwhile\ #1}
\title{Optimal Offline and Competitive Online Strategies for Transmitter-Receiver Energy Harvesting}

\author{Rushil Nagda, Siddharth Satpathi, Rahul Vaze
\thanks{Rahul Vaze's research is supported by ITRA grant 13X300.}
}

\begin{document}
\maketitle
\thispagestyle{empty}
\pagestyle{empty}
\begin{abstract}
Transmitter-receiver energy harvesting model is assumed, where both the transmitter and receiver are powered by random energy source. Given a fixed number of bits, the problem is to find the optimal transmission power profile at the transmitter and ON-OFF profile at the receiver to minimize the transmission time. Structure of the optimal offline strategy is derived together with an optimal offline policy. An online policy with competitive ratio of strictly less than two is also derived.
\end{abstract}

\begin{IEEEkeywords}Energy harvesting, offline algorithm, online algorithm, competitive ratio.
\end{IEEEkeywords}
\section{Introduction}
Extracting energy from nature to power communication devices has been an emerging area of research.  Starting with \cite{ozel2012achieving, SharmaEH2014}, a lot of work has been reported on finding the capacity, approximate capacity \cite{dong2014near}, structure of optimal policies \cite{sinha2012optimal}, optimal power transmission profile  \cite{UlukusEH2011b, UlukusEH2011c, michelusi2012optimal, VazeEHICASSP14}, competitive online algorithms \cite{VazeEH2011}, etc. One thing that is common to almost all the prior work is the assumption that energy is harvested only at the transmitter while the receiver has some conventional power source. This is clearly a limitation, however, helped to get some critical insights into the problem.

In this paper, we broaden the horizon, and study the more general problem when energy harvesting is employed both at the transmitter and the receiver. 
The joint (tx-rx) energy harvesting model has not been studied in detail and only some preliminary results are available, e.g., a constant approximation to the maximum throughput has been derived in \cite{VazeEH2014}.
This problem is fundamentally different than using energy harvesting only at the transmitter, where receiver is always assumed to have energy to receive. 
The receiver energy consumption model is binary, since it uses a fixed amount of energy to stay {\it on}, and is {\it off} otherwise. Since useful transmission happens only when the receiver is \textit{on}, the problem is to find jointly optimal decisions about transmit power and receiver ON-OFF schedule. Under this model, there is an issue of coordination between the transmitter and receiver to implement the joint decisions, however, we ignore that currently in the interest to make some analytical progress.

We study the canonical problem of finding the optimal transmission power and receiver ON-OFF schedule to minimize the time required for transmitting a fixed number of bits. 
We first consider the offline case, where the energy arrivals both at the transmitter and the receiver are assumed to be known non-causally. Even though offline scenario is unrealistic, it still gives some design insights. 
Then we consider the more useful online scenario, where both the transmitter and receiver only have causal information about the energy arrivals. To characterize the performance of an online algorithm, typically, the metric of competitive ratio is used that is defined as the maximum ratio of profit of the online and the offline algorithm over all possible inputs. 


In prior work \cite{UlukusEH2011b}, an optimal offline algorithm has been derived for the case when energy is harvested only at the transmitter, which cannot be generalized with energy harvesting at the receiver together with the transmitter. To understand  the difficulty, assume that the receiver can be \textit{on} for a maximum time of $T$.
The policy of \cite{UlukusEH2011b} starts transmission at time $0$, and power transmission profile is the
one that yields the tightest piecewise linear energy consumption curve that lies under the energy harvesting cure at all times and touches the energy harvesting curve at end time.
With receiver {\it on} time constraint, however, the policy of \cite{UlukusEH2011b} may take more than $T$ time and hence may not be feasible.
So, we may have to either delay the start of transmission and/or keep stopping in-between  to accumulate more energy to transmit with higher power for shorter bursts, such that the total time for which transmitter and receiver is \textit{on}, is less than $T$.

%
%
%
%

The contributions of this paper are : 
\begin{itemize}
\item For the offline scenario, we derive the structure of the optimal algorithm, and then propose an algorithm that is shown to satisfy the optimal structure. The power profile of the proposed algorithm is fundamentally different than the optimal offline algorithm of \cite{UlukusEH2011b}, 
however, the two algorithms have some common structural properties.
\item For the online scenario, we propose an online algorithm and show that its competitive ratio is strictly less than $2$ for any energy arrival inputs. With only energy harvesting at the transmitter, a $2$-competitive online algorithm has been derived in \cite{VazeEH2011}. This result is more general with different proof technique that allows energy harvesting at the receiver. 
\end{itemize}

\section{System Model}
The energy arrival instants at transmitter are marked by $\tau_i$'s with energy $\ETx_i$'s for $i \in \{0,1,..\}$. The total energy harvested at the transmitter till time $t$ is given by $\ETx(t)=\sum\limits_{i:\tau_i < t}\ETx(t)$. Similarly, the energy arrival instants at the receiver are denoted as $r_i$ with energy $\cR_i$. With fixed power consumption of $P_{r}$ at the receiver to stay ${\it on}$, each energy arrival of $\cR_i$ adds $\Gamma_i= \frac{\cR_i}{P_{r}}$amount of receiver {\it on} time, and the total `time' harvested at the receiver till time $t$ is given by $\TRx(t)$.

 


Assuming an AWGN channel,  the rate of bits transmission, using transmit power $p$ and receiver is {\it on} is given by a monotonically increasing function $g(p)$, such that, $g(0)=0\text{ and }\lim_{x\rightarrow \infty} g(x)= \infty$, $g(x)\text{ is concave}$, $\frac{g(x)}{x} \text{ is convex monotonically decreasing} $, and $ \lim_{x\rightarrow \infty} \frac{g(x)}{x}= 0$. $\log$ function is one such example.

Let a transmission policy change its transmission power at time instants $s_i$'s, i.e. $p_i$ is the power used between time $s_i$ and $s_{i+1}$, and $p_i \ne p_{i+1}$. The start and the end time of any policy is denoted by $s_1$ and $s_{N+1}$, respectively. Thus, any policy can be represented as $\{\bm{p}$, $\bm{s}, N\}$, where $\bm{p}=\{p_1, p_2, .., p_N\}$ and $\bm{s}=\{s_1, s_2, .., s_{N+1}\}$. The energy used by a policy at the transmitter upto time $t$ is denoted by $U(t)$, and the number of bits sent by time $t$ is represented by $B(t)$. Clearly, for $j=\argmax_{i}\{\tau_i < t\}$, 
$U(t)=\sum_{i=1}^{j} p_i(s_{i+1}-s_i)+p_{j+1}(t-s_j)$, and $
B(t)=\sum_{i=1}^{j} g(p_i)(s_{i+1}-s_i)+g(p_{j+1})(t-s_j)$. Similarly, the total time for which the receiver is {\it on} till time $t$ is denoted as $O(t)$.

We assume that an infinite battery capacity is available both at the transmitter and the receiver to store the harvested energy. Finite battery case can be handled, however, the description is more laborious and currently under preparation. Our objective is, given a fixed number of bits $B_0$, minimize the time of their transmission. For any policy, the total time for which the receiver is \textit{on} is referred to as the `transmission time' or the `transmission duration',  and the time by which the transmission of $B_0$ bits is finished, is called as the `finish time'.
Thus, we want to minimize the finish time, 
\begin{align}
&\min_{\{\textbf{p},\textbf{s},N\}}			&& T\label{pb1}
\\
&\text{subject to} 				&& B(T)=B_0, 
\label{pb1_constraint_bits}
\\
&     										&& U(t)\le \mathcal{E}(t),  		&&& \forall \; t\;\in\;[0,T], \label{pb1_constraint_energy}
\\
&    										&& O(t) \le \Gamma(t).
\label{pb1_constraint_time}
\end{align}
Constraints \eqref{pb1_constraint_energy} and \eqref{pb1_constraint_time} are the energy neutrality constraints at the transmitter and receiver, i.e.  energy/time used cannot be more than available energy/time. 
Compared to the no receiver constraint \cite{UlukusEH2011b}, problem \eqref{pb1} is far more complicated, since it involves jointly solving for optimal transmitter power allocation and time for which to keep the receiver {\it on}.

\section{OPTIMAL OFFLINE ALGORITHM}

In this section, we consider an offline scenario, i.e., all energy arrival epochs $\tau_i$'s at the transmitter are known ahead of time non-causally. Moreover, for simpler exposition, however, without losing the richness of the problem, we assume that the receiver gets energy $\cR$ only at time $0$, and hence the total receiver {\it on} time is $\Gamma_0= \frac{\cR}{P_r}$.
With only one receiver arrival, constraint \eqref{pb1_constraint_time} in Problem \eqref{pb1} specializes to $\sum_{i=1:p_i\neq 0}^{N}(s_{i+1}-s_i)\le \TRx_0$. Note that even with restriction, the problem is still challenging since we have to find the optimal receiver {\it on} periods (breakup of the total receiver {\it on } time of $\TRx_0$) depending on the energy arrivals at the transmitter to minimize the finish time.

\begin{lemma}
In an optimal solution to \eqref{pb1}, if $p_i\neq 0$, then $p_i\ge p_j$ $\ \forall \ j<i$ with $i,j\in \{1,2..N\}$.
\label{lemma_increasing_power}
\end{lemma}
Proof involves the argument that, if powers are decreasing, then utilizing the concavity of $g(p)$, we can construct another strategy that can send same number of bits in less time. It is similar to Lemma 1 in \cite{UlukusEH2011b}, however, requires a separate proof because, with the receiver {\it on} time constraint, the optimal solution can intermittently have zero transmit powers. 
Note: For space constraints, proofs are included/omitted depending on their significance and the non-triviality.

\begin{lemma}
The optimal solution to \eqref{pb1} may not be unique, but there always exists an optimal solution where once transmission has started, the receiver remains `\textit{on}' throughtout, until the transmission is complete. \label{lemma_nobreaks}
\end{lemma}
Lemma \ref{lemma_nobreaks} tells us that there is no need to stop in-between transmission and start again. Without affecting optimality, the start of the transmission can be delayed so that transmission power is non-zero throughout. 
\begin{proof}
We construct an optimal solution for which $p_i>0$ for all $i\in\{1,..,N\}$, i.e., with no breaks in transmission, from any other optimal solution. Let an optimal policy $X$ be characterized by $\{\bm{p},\bm{s},N\}$. Now, if $p_i\neq 0\; \forall \ i$, then we are done. Suppose some powers, say $p_{i_1},p_{i_2},...,p_{i_k}=0$ for some $k<N$, where $i_1<i_2<..<i_k$. We first look at instant $i_1$.

Consider Fig. \ref{fig_Lemma2} (a), and a new policy (say $Y$) which is same as policy $X$ before time $s_{i_1-1}$ and after time $s_{i_1+1}$. But, it keeps the receiver \textit{off} for a duration of $(s_{i_1+1}-s_{i_1})$ starting from time $s_{i_1-1}$ (i.e. from $s_{i_1-1}$ to $s_{i_1}'=(s_{i_1-1}+s_{i_1+1}-s_{i_1})$) and transmits with power $p_{i_1-1}$ from time $s_{i_1}'$ till $s_{i_1+1}$. $Y$ transmits same amount of bits in same time as $X$ and also satisfies constraints \eqref{pb1_constraint_bits}-\eqref{pb1_constraint_time}. So $Y$ is also an optimal policy. But the receiver \textit{off} duration in $Y$, $(s_{{i_1+1}}-s_{i_1})$, has been shifted to left. 

Next, we generate another policy $Z$ from $Y$ by shifting the \textit{off} duration $s_{i_1}'-s_{i_1-1}=(s_{{i_1+1}}-s_{i_1})$ to start from epoch $s_{i_1-2}$ upto $s_{i_1-1}'$, $s_{i_1-1}'-s_{i_1-2}=s_{i_1}'-s_{i_1-1}=(s_{{i_1+1}}-s_{i_1})$, as shown Fig. \ref{fig_Lemma2} (b). $p_{i_1-2}$  is shifted right to start from $s_{i_1-1}'$. Note that $Z$ is also optimal. We continue this process of shifting the receiver \textit{off} period to the left to generate new optimal policies till we reach a policy (say $W$) where the receiver is \textit{off} for time $(s_{{i_1+1}}-s_{i_1})$ from $s_1$, i.e. from $s_{1}$ to $s_1'$, $s_1'-s_1=(s_{{i_1+1}}-s_{i_1})$, as shown in Fig. \ref{fig_Lemma2}(c). As $W$ has $0$ transmission power from the start time $s_1$ to $s_1'$, the effective start time of $W$ can now be changed to $s_1'$. 

We can repeat this procedure for each \textit{off} period corresponding to $p_{i_2},...,p_{i_k}$ till the total \textit{off} period is shifted to the beginning of transmission. 
This results in a policy with no zero powers in between, that starts \textit{after} time $s_1$ (at $s_1+(s_{{i_1+1}}-s_{i_1})+..+(s_{{i_k+1}}-s_{i_k})$) and ends at the same time $s_{N+1}$ as policy $X$.
\end{proof}
\begin{figure}[htb]
  \centering
  \centerline{\includegraphics[width=8cm]{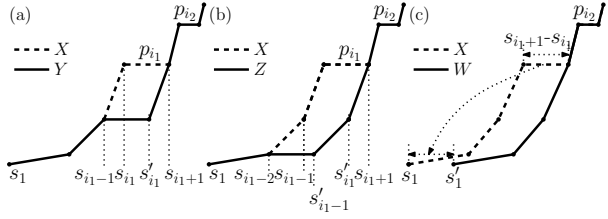}}
\caption{Illustration of Lemma \ref{lemma_nobreaks}. Receiver \textit{off} time of $(s_{j}-s_{i_1})$ is progressively shifted to left as shown in (a) to (b) to (c).}\label{fig_Lemma2}
\end{figure}
\textit{In the subsequent discussion, the optimal solution means one with no breaks in transmission.}
\begin{lemma}
For optimal policy $\{\bm{p},\bm{s},N\}$, $s_i=\tau_j$ for some $j$, $U(s_i)=\ETx(s_i^-)\ \forall i\in\{2,..,N\}$, and $U(s_{N+1})=\ETx(s^-_{N+1})$.
\label{lemma_energy_consumed} 
\end{lemma}
\begin{proof}
By Lemma \ref{lemma_increasing_power} and \ref{lemma_nobreaks},  $p_i\neq 0$ and $p_{i+1}\ge p_i,\forall 1\le i\le N$. So, the proof follows similar to Lemma 2,3 in \cite{UlukusEH2011b}. 
\end{proof}
Lemma \ref{lemma_energy_consumed} states that in an optimal solution, the transmission power changes only at energy arrival epochs, and the energy used is equal to all the energy that has arrived till then.
It may happen that at some epoch $\tau_k$, $U(\tau_k) = \ETx(\tau_k^-)$ holds true, but transmission power does not change. For notational simplicity, we inculde all such $\tau_k$'s in $\bm{s}$, where $U(\tau_k)= \ETx(\tau_k^-)$.
  
\begin{lemma}
Consider two policies $X$, $\{\bm{p},\bm{s},N\}$  and $Y$, $\{\bm{\widetilde{p}},\bm{\widetilde{s}},N\}$, which are feasible with respect to energy constraint \eqref{pb1_constraint_energy}, have non-decreasing powers and transmit same number of bits in total. If $Y$ is same as $X$ from time $s_2$ to $s_{N}$, but $\widetilde{p}_1=p_1-\alpha,\widetilde{p}_N=p_N+\beta, \widetilde{s}_1=s_1-\gamma, \widetilde{s}_{N+1}=s_{N+1}-\delta $ and $U(s_{N+1})=U(\widetilde{s}_{N+1})$, where $\alpha ,\beta,\gamma,\delta>0$, then $
(\widetilde{s}_{N+1}-\widetilde{s}_1)>(s_{N+1}-s_1)$.
\label{lemma_increase_time}
\end{lemma}
%
This lemma states that if we take any feasible policy, and decrease its first power $p_1$ \& increase its last power $p_N$ while keeping the same number of transmitted bits, the time of transmission will increase, while the finish time of the policy will reduce. The proof is algebraic using the concavity of $g(p)$, and convexity of $g(p)/p$.

\begin{lemma}  For an optimal policy $\{\bm{p},\bm{s},N\}$, either $s_{N+1} - s_1 = \TRx_0$ or $s_1 = 0$ . 
\label{transmission_duration}
\end{lemma}
\begin{proof} We use the method of contradiction.
Suppose the optimal policy say $X$, starts at $s_1>0$ and has transmission time $(s_{N+1}-s_1)< \TRx_0$. 
We will generate another policy which has finish time less than that of $X$, having transmission time squeezed in between $(s_{N+1}-s_1)$ and $\TRx_0$.
Consider policy $Y$ ($\{\bm{\widetilde{p}},\bm{\widetilde{s}},N\}$) in relation to $X$, as defined in Lemma \ref{lemma_increase_time}. As $\alpha$, $\beta$, $\delta$, $\gamma$ are all related (by constraints presented in Lemma \ref{lemma_increase_time}), choice of one variable (we consider $\alpha$) defines $Y$. By definition of $s_i$'s, $s_{2}$ is the first energy arrival which is on the boundary of energy constraint (\ref{pb1_constraint_energy}) i.e. $U(s_2)=\ETx(s_2^-)$ and $s_{N}$ is the last epoch satisfying $U(s_N)=\ETx(s_N^-)$. Hence, we can choose $\alpha>0$, such that $\widetilde{p}_1$ and $\widetilde{p}_N$ would be feasible with respect to energy constraint (\ref{pb1_constraint_energy}). Note that if $s_1=0$, then any value of $\alpha$ would have made $\widetilde{p}_1$ infeasible. 

From Lemma \ref{lemma_increase_time}, we know that the transmission time of policy $Y$ is more than that of $X$, i.e. $(\widetilde{s}_{N+1}-\widetilde{s}_1) > (s_{N+1} - s_{1})$. From the hypothesis $(s_{N+1}-s_1) < \TRx_0$. Therefore, let $(s_{N+1}-s_1)=\TRx_0-\epsilon$, with $\epsilon >0$. If the chosen value of $\alpha$ is such that $\gamma -\delta\le\epsilon$, then $\left( \widetilde{s}_{N+1}-\widetilde{s}_1\right)<\TRx_0$. If not, then we can further reduce $\alpha$ so that $\gamma -\delta\le\epsilon$ ($\alpha$,$\beta$,$\gamma$,$\delta$ being related by continuous functions).  Note that, when $\epsilon=0$, any choice of $\alpha$ would make $\left( \widetilde{s}_{N+1}-\widetilde{s}_1\right)>\TRx_0$. Hence, with this choice of $\alpha$, $(s_{N+1}-s_1)<\left( \widetilde{s}_{N+1}-\widetilde{s}_1\right)<\TRx_0$ holds and  policy $Y$ 
contradicts the optimality of policy $X$ (as finish time of $Y$, $\widetilde{s}_{N+1}=s_{N+1}-\delta <s_{N+1}$ from Lemma \ref{lemma_increase_time}). Thus $s_{N+1}-s_1=\TRx_0$ if $s_1\neq 0$ in an optimal policy.
\end{proof}

%
\begin{theorem}
A policy $\{\bm{p},\bm{s},N\}$ is an optimal solution to Problem 1 if and only if, 
\label{th_algo1_1}
\begin{align}
&\sum_{i=1}^{i=N}g(p_i)(s_{i+1}-s_i)=B_0; 								
\label{claim1}
\\
&p_1\le p_2 \ldots \le p_N;
\label{claim3}  
\\
&\nonumber s_i=\tau_j  \ \ \ \ \ \ \ \ \ \ \ \ \ \ \ \text{ for some } j, i\in \{2,..,N\} \ \text{ and }
\\
& U(s_i)=\ETx(s_i^-), \ \ \ \ \ \forall i\in \{2,..,N+1\};
\label{claim4}
\\
&\nonumber s_{N+1}-s_1=\TRx_0,  \ \ \ 						\text{ if } s_1>0 \text{ or }
\\
& s_{N+1}\le \TRx_0,			\ \ \ \ \ \ \ \ \				\text{ if } s_1=0;
\label{claim2}
\\
&\exists s_j:s_j\in \bm{s} \text{ and } s_j=\tau_q,
\label{claim5}
\end{align}
where $\tau_q$ is defined in INIT\_POLICY of section IV.
\end{theorem}
\begin{proof} The necessity of these conditions is established in Lemmas \ref{lemma_increasing_power}-\ref{lemma_Q}. For lack of space, sufficiency proof is omitted.
\end{proof}
\section{Optimal Offline Algorithm }
In this section, we propose an offline algorithm $\mathsf{OFF}$, and show that it satisfies the sufficiency conditions of Theorem \ref{th_algo1_1}.
Algorithm $\mathsf{OFF}$ first finds an initial feasible solution via INIT\_POLICY, and then iteratively improves upon it via PULL\_BACK. Finally, QUIT produces the output.
\subsection{INIT\_POLICY} 
We find a simple constant power policy that is feasible and starts as early as possible. Also, we try to make it satisfy most of the sufficient conditions of Theorem \ref{th_algo1_1}.

\textit{Step1:} Identify the first energy arrival instant $\tau_n$, so that using $\ETx(\tau_n)$ energy and $\TRx_0$ time, $B_0$ or more bits can be transmitted with a constant power (say $p_c$), i.e. $\TRx_0g\left(\dfrac{\ETx(\tau_n)}{\TRx_0}\right)\ge B_0$. Then solve for $\widetilde{\TRx}_0$,
\begin{small}
\begin{equation}
\widetilde{\TRx}_0\, g\left(\dfrac{\ETx(\tau_n)}{\widetilde{\TRx}_0}\right)= B_0,\ p_c = \dfrac{\ETx({\tau_n})}{\widetilde{\TRx}_0}.
\label{INIT_POLICY_time}
\end{equation}
\end{small}
\begin{figure}
\centering
  \centerline{\includegraphics[width=8cm]{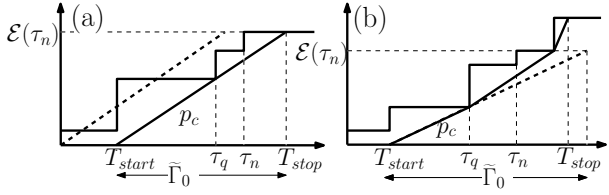}}
\caption{Figure showing point $\tau_q$.}\label{straight}
\end{figure} 
\textit{Step2:} Find the earliest time $T_{start}$, such that transmission with power $p_c$ from $T_{start}$ for $\widetilde{\TRx}_0$ time, is feasible with energy constraint \eqref{pb1_constraint_energy}. Set $T_{stop} = T_{start} + \widetilde{\TRx}_0$. Let $\tau_q$ be the \textit{first epoch} where $U(\tau_q) = \ETx(\tau_q^-)$ (Fig. \ref{straight}).
Next Lemma shows that point $\tau_q$ thus found is a `good' starting solution.
\begin{lemma}
In every optimal solution, at energy arrival epoch $\tau_q$ defined in INIT\_POLICY, $U(\tau_q)=\ETx(\tau_q^-)$.
\label{lemma_Q}
\end{lemma}
Continuing with INIT\_POLICY, if $U(T_{stop}) = \ETx(T_{stop}^-)$ as shown in Fig. \ref{straight}(a), then terminate INIT\_POLICY with constant power policy $p_c$. 

Otherwise if $U(T_{stop}) < \ETx(T_{stop}^-)$, then 
modify the transmission after $\tau_q$ as follows. Set $\widetilde{B}_0 = (T_{stop} - \tau_q)g(p_c)$,  which denotes the number of bits left to be sent after time $\tau_q$. Then apply Algorithm 1 of \cite{UlukusEH2011b} \textit{from time $\tau_q$} to transmit $\widetilde{B}_0$ bits in as minimum time as possible without considering the receiver {\it on} time constraint. 
Update $T_{stop}$, to where this policy ends. So, $U(T_{stop}) = \ETx(T_{stop}^-)$ from \cite{UlukusEH2011b}. Since Algorithm 1 \cite{UlukusEH2011b} is optimal, it takes minimum time ($=T_{stop}-\tau_q$) to transmit $\widetilde{B}_0$ starting at time $\tau_q$. 
However, using power $p_c$ to transmit $\widetilde{B}_0$ takes $(T_{start}+\widetilde{\TRx}_0 - \tau_q)$ time.
Hence, $T_{stop}\le (T_{start}+\widetilde{\TRx}_0)$. 
As $\widetilde{\TRx}_0\le \TRx_0$ from \eqref{INIT_POLICY_time}, $(T_{stop}- T_{start})\le \TRx_0$. This shows that solution thus found using Algorithm 1 \cite{UlukusEH2011b}, is indeed feasible with receiver time constraint \eqref{pb1_constraint_time}. Now, output of INIT\_POLICY is a policy that transmits at power $p_c$ from $T_{start}$ to $\tau_q$, and after $\tau_q$ uses Algorithm 1 of \cite{UlukusEH2011b}.  

\begin{figure}
\centering
  \centerline{\includegraphics[width=8cm]{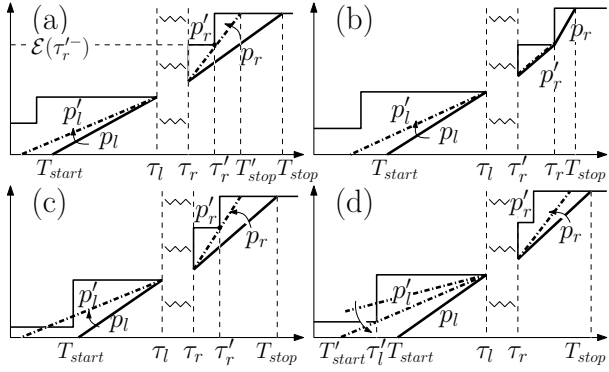}}
\caption{Figures showing possible configurations in any iteration of the PULL\_BACK. The solid line represents the transmission policy in the previous iteration and dash dotted lines are for the current iteration.}\label{figure_Algorithm1}
\end{figure}
\subsection{ PULL\_BACK}
Now, we describe the iterative subroutine PULL\_BACK whose input is policy $\{\bm{p},\bm{s},N\}$ output by INIT\_POLICY.
Clearly $\{\bm{p},\bm{s},N\}$ satisfies all but structure \eqref{claim2} of Theorem \ref{th_algo1_1}. 
So, the main idea of PULL\_BACK is to increase the transmission duration from $(s_{N+1}-s_1)\le \widetilde{\TRx}_0$, in INIT\_POLICY, to $\TRx_0$ in order to satisfy \eqref{claim2}, while decreasing the finish time for reaching the optimal	 solution. To achieve this, we utilize the structure presented in Lemma \ref{lemma_increase_time} and iteratively increase the last transmission power $p_N$, and decease the first transmission power $p_1$. 
%

Initialize $\tau_l=s_2,\tau_r=s_N,p_l=p_1,p_r=p_N,T_{start}=s_1,T_{stop}=s_{N+1}$. In any iteration, $\tau_{l}$ and $\tau_{r}$ are assigned to the first and last energy arrival epochs, where $U(\tau_l)=\ETx(\tau_l^-)$ and $U(\tau_r)=\ETx(\tau_r^-)$. $p_l$ and $p_r$ are the constant transmission powers before $\tau_l$ and after $\tau_r$, respectively. 
We reuse the notation $\tau$ here, because $\tau_l$ and $\tau_r$ will occur at energy arrival epochs from Lemma \ref{lemma_energy_consumed}. 
$T_{start}$ and $T_{stop}$ are the start and finish time of the policy, found in any iteration. $\tau_l, \tau_r, p_l, p_r, T_{start}, T_{stop}$ get updated to $\tau_l', \tau_r', p_l', p_r', T'_{start}, T'_{stop}$ over an iteration. In any iteration, only one of $\tau_l$ or $\tau_r$ gets updated, i.e., either $\tau_l'=\tau_l$ or $\tau_r'=\tau_r$. Further, PULL\_BACK ensures that \textit{transmission powers between $\tau_l$ and $\tau_r$ do not get changed} over an iteration.
Fig. \ref{figure_Algorithm1} shows the possible updates in an iteration of PULL\_BACK.

\textit{Step1, Updation of $\tau_r$, $p_r$:} Initialize $p_r'=p_r$ and increase $p_r'$ till it hits the boundary of energy constraint \eqref{pb1_constraint_energy}, say at $(t_r',\ETx(t_r'^-))$ as shown in Fig. \ref{figure_Algorithm1}(a). The last epoch where $p_r'$ hits \eqref{pb1_constraint_energy} is set to $\tau_r'$. So, $U(\tau_r') = \ETx(\tau_r'^-)$. Set $T_{stop}'$ to where power $p_r'$ ends.
Calculate $p_l'$ such that decrease in bits transmitted due to change from $p_r$ to $p_r'$ is compensated by increasing $p_l$ to $p_l'$, via
\begin{align}
\nonumber g(p_r)(T_{stop}&-\tau_r)-g(p_r')(T_{stop}'-\tau_r')
\\
&=g(p_l')\frac{\ETx(\tau_l'^-)}{p_l'}-g(p_l)(\tau_l-T_{start}).
\label{eq_example1}
\end{align}
Suppose, $p_r$ can be increased till infinity without violating \eqref{pb1_constraint_energy}, as shown in Fig. \ref{figure_Algorithm1}(b).  This happens when there in no energy arrival between $\tau_r$ and $T_{stop}$. In this case, set $p_r'$ to the transmission power at $\tau_r^-$. Set $\tau_r'$ as the epoch where $p_r'$ starts, and $T_{stop}'$ to $\tau_r$. Calculate $p_l'$ similar to \eqref{eq_example1}.

\textit{Step2, Updation of $\tau_l, p_l$}: If $p_l'$ obtained from \textit{Step1} is feasible, as shown in Fig. \ref{figure_Algorithm1}(a), set $T_{start}'=\tau_l-\frac{\ETx(\tau_l'^-)}{p_l'}$, $\tau_l'=\tau_l$. Proceed to \textit{Step3}. Otherwise, if $p_l'$ is not feasible, as shown in Fig. \ref{figure_Algorithm1}(c), the changes made to $\tau_r',p_r'$ in \textit{Step1} are discarded. As shown in Fig. \ref{figure_Algorithm1} (d), $p_l'$ is increased from its value in \textit{Step1} until it becomes feasible. $\tau_l'$ is set to the first epoch where $U(\tau_l') = \ETx(\tau_l'^-)$. Similar to \textit{Step1}, calculate $p_r'$ such that the increase in bits transmitted due to change of $p_l$ to $p_l'$ is compensated, and update $T_{stop}'$ accordingly. Set $\tau_r'=\tau_r$. Proceed to \textit{Step3}.

\textit{Step3, Termination condition:} If $T_{stop}' - T_{start}' \ge \TRx_0$ or $T_{start}' = 0$, then terminate PULL\_BACK. Otherwise, update $\tau_l, \tau_r, p_l, p_r, T_{start}, T_{stop}$ to $\tau_l', \tau_r', p_l', p_r', T'_{start}, T'_{stop}$ receptively and GOTO \textit{Step1}.
\begin{lemma}
Transmission time $(T_{stop}-T_{start})$ monotonically increases over each iteration of PULL\_BACK.
\label{lemma_PULL_BACK_power}
\end{lemma}
\begin{theorem} Worst case running time of PULL\_BACK is linear with respect to the number of energy harvests before finish time of INIT\_POLICY.
\end{theorem}
\begin{proof}
%
%
%
%
%
Since, in an iteration of PULL\_BACK, either $\tau_r$ or $\tau_l$ updates, the number of iterations is bounded by the values attained by $\tau_l$, plus that of $\tau_r$. Initially, $\tau_l\le \tau_q$ and $\tau_r \ge \tau_q$.
As $\tau_l$ is non-increasing across iterations, $\tau_l\le \tau_q$ throughout.
Assume that $\tau_r$ remains $\ge \tau_q$ across INIT\_POLICY. Then, both $\tau_l$ and $\tau_r$ can at max attain all $\tau_i$'s less than finish time of initial feasible policy. Hence, we are done.

Now, it remains to show that $\tau_r\ge \tau_q$. $\tau_n$ is defined as the first energy arrival epoch with which $B_0$ or more bits can be transmitted in $\TRx_0$ time and $\tau_q\le \tau_n$, by definition. 
So, when $T_{stop}$ becomes $\le \tau_n \, or \,  \tau_q$, then transmission time, $(T_{stop}-T_{start})$, should be  $>\TRx_0$. 
But, in the initial iteration $(T_{stop}-T_{start})\le \TRx_0$ and $(T_{stop}-T_{start})$ increases monotonically, from Lemma \ref{lemma_PULL_BACK_power}. Hence, PULL\_BACK will terminate before $T_{stop}$ (and therefore $\tau_r$) decreases beyond $\tau_q$.
\end{proof}
\subsection{QUIT}If $T_{start}' = 0$ and $T_{stop}' - T_{start}' \le \TRx_0$ upon PULL\_BACK's termination, then PULL\_BACK's policy at termination is output. Note that structure \eqref{claim2} holds for this policy. Otherwise, if $T_{stop}' - T_{start}' > \TRx_0$ (which happens for the first time), then we know that in penultimate step $T_{stop} - T_{start} < \TRx_0$.
Hence, we are looking for a policy that starts in $ [T_{start} ,\ T_{start}']$ and ends in $[T_{stop} ,\ T_{stop}']$, whose transmission time is equal to $\TRx_0$. 
Hence, we solve for $x,y$ (let the solution be $\hat{x},\hat{y}$),
\begin{align}
\nonumber (\tau_l-x)& \; g\left(\frac{\ETx(\tau_l^-)}{\tau_l-x}\right)+(y-\tau_r)\; g\left(\frac{\ETx(T_{stop}^-)}{y-\tau_r}\right)\\
&=g(p_l)(\tau_l-T_{start})+g(p_r)(T_{stop}-\tau_r),
\label{eq_termination_0}
\\
y-x&=\TRx_0.
\label{eq_termination}
\end{align}
At penultimate iteration, $(x,y)=(T_{start},T_{stop})$, \eqref{eq_termination_0} is satisfied and $y-x<\TRx_0$.
At $(x,y)=(T_{start}',T_{stop}')$, as $\ETx(T_{stop}^-)=\ETx(T_{stop}'^-)$, \eqref{eq_termination_0} is satisfied and $y-x>\TRx_0$. 
So, there must exist a solution $(\hat{x},\hat{y})$ to \eqref{eq_termination_0}, where $\hat{x}\in [T_{start}',T_{start}]$, $\hat{y}\in [T_{stop}',T_{stop}]$ and $\hat{y}-\hat{x}=\TRx_0$, for which, \eqref{claim2} holds. Output with this policy which starts at $\hat{x}$ and ends at $\hat{y}$.

\begin{theorem}
The transmission policy proposed by Algorithm $\mathsf{OFF}$ is an optimal solution to Problem \eqref{pb1}.
\label{th_algo1_2}
\end{theorem}
\begin{proof}
We show that  Algorithm $\mathsf{OFF}$ satisfies the sufficiency conditions of  Theorem \ref{th_algo1_1}.
To begin with, we prove that the power allocations satisfy \eqref{claim3} by induction. 
First we establish the base case that INIT\_POLICY's output satisfies \eqref{claim3}.
If INIT\_POLICY returns the constant power policy $p_c$ from time $T_{start}$ to $T_{stop}$, then clearly the claim holds. 

Otherwise, INIT\_POLICY applies Algorithm 1 from \cite{UlukusEH2011b} with $\widetilde{B}=B_0-g(p_c)(\tau_q-T_{start})$ bits to transmit after time $\tau_q$. 
Algorithm 1 from \cite{UlukusEH2011b} ensures that transmission powers are non-decreasing after $\tau_q$. 
So we only prove that the transmission power $p_c$ between time $T_{start}$ and $\tau_q$ is $\le$ to the transmission power just after $\tau_q$ (say $p_q$), via contradiction. Assume that $p_q<p_c$. Let transmission with $p_q$ end at an epoch $\tau_{q'}$, where $U(\tau_{q'})=\ETx(\tau_{q'}^-)$ form \cite{UlukusEH2011b}.  The energy consumed between time $\tau_q$ to $\tau_{q'}$ with power $p_c$ is,
%
\begin{equation} 
p_c(\tau_{q'}-\tau_q)>p_q(\tau_{q'}-\tau_q)\stackrel{(a)}{=}(\ETx(\tau_{q'}^-)-\ETx(\tau_q^-)),\label{eq_1_algo1_modified}
\end{equation} 
where $(a)$ follows from $U(\tau_q)=\ETx(\tau_q^-)$. Further, the maximum amount of energy available for transmission between $\tau_q$ and $\tau_{q'}$ is $\left(\ETx(\tau_{q'}^-)-\ETx(\tau_q^-)\right)$. By \eqref{eq_1_algo1_modified}, transmission with $p_c$ uses more than this energy and therefore it is infeasible between time $\tau_q$ and $\tau_{q'}$. But, by definition of $p_c$, transmission with power $p_c$ is feasible till time $(T_{start}+\widetilde{\TRx}_0)$. Also, $\tau_{q'}\le T_{stop}$ by definition of $\tau_{q'}$ and $T_{stop}\le (T_{start}+\widetilde{\TRx}_0)$. So, power $p_c$ must be feasible till $\tau_{q'}$ and we reach a contradiction.        

Now, we assume that the transmission powers from PULL\_BACK are non-decreasing till its $n^{th}$ iteration. Therefore, as transmission powers between $\tau_l$ and $\tau_r$ does not change over an iteration, powers would remain non-decreasing in the $(n+1)^{th}$ iteration if we show that $p_l'<p_l$ and $p_r'>p_r$. 
In any iteration, by  definition, either $\tau_l$ or $\tau_{r}$ updates. Assume $\tau_l$ gets updated to $\tau_{l}'$, $p_l$ to $p_l'$, $p_r$ to $p_r'$ and $\tau_r$ remains same, shown Fig. \ref{figure_Algorithm1}(d) (when $\tau_r$ updates, the proof follows similarly).  
Then we are certain that $p_{r}'>p_r$ by algorithmic steps. So from $n^{th}$ to $(n+1)^{th}$ iteration, the number of bits transmitted after $\tau_r$ should decrease. Thus, the number of bits transmitted before $\tau_l$ must be increasing. This implies $p_l'\le p_l$. Hence, transmission powers by output by $\mathsf{OFF}$ are non-deceasing and it satisfies \eqref{claim3}.

%



Now consider structure \eqref{claim5}. As $\tau_q$ is present in INIT\_POLICY, the only way 
it cannot be part of the policy in an iteration of PULL\_BACK is when $\tau_r$ decreases beyond $\tau_q$. But $\tau_r\ge \tau_q$ as shown in Theorem 2.  So, the policy output by $\mathsf{OFF}$ includes $\tau_q$.
By arguments presented at end of OUIT, we know that $\mathsf{OFF}$ satisfies \eqref{claim2}.
To conclude, $\mathsf{OFF}$ satisfies \eqref{claim1}-\eqref{claim5}, hence is an optimal algorithm.
\end{proof}

\section{ONLINE ALGORITHM}
In this section, we consider solving Problem \eqref{pb1}  in the more realistic online scenario, where the transmitter and the receiver 
are assumed to have only causal information about energy arrivals. To consider the most general model, even the  distribution of future energy arrivals is unknown at both the transmitter and the receiver. Moreover, we do not limit ourselves to just one energy arrival at receiver as done for the offline case. 

\textit{Notation:} Let $B_{\mbox{\scriptsize{rem}}}(t)$ and $E_{\mbox{\scriptsize{rem}}}(t)$ denote the remaining number of bits and energy left at transmitter at any time $t$, respectively for the online algorithm. In place of $\{\bm{p},\bm{s},N\}$ for the offline case, we use the notation $\{\bm{l},\bm{b},M\}$ to denote an online policy, with identical definitions. $T_{\mbox{\scriptsize{online}}}$ and $T_{\mbox{\scriptsize{off}}}$ represent the finish time of the online and the optimal offline algorithm to Problem \eqref{pb1}, respectively.
We use the competitive ratio as a metric where we say that an online algorithm is $r$-competitive, if over all possible energy arrivals at the transmitter and the receiver, the ratio of 
$T_{\mbox{\scriptsize{online}}}$ to $T_{\mbox{\scriptsize{off}}}$ is bounded by $r$, i.e., $\displaystyle\max_{\ETx(t),\TRx(t)\hspace{0.5mm} \forall t}\dfrac{T_{\mbox{\scriptsize{online}}}}{T_{\mbox{\scriptsize{off}}}}\le r$.

\textit{Online Algorithm:} The algorithm waits till time $T_{\mbox{\scriptsize{start}}}$ 
which is the earliest energy arrival at transmitter or time addition at receiver such that using the energy $\ETx(T_{\mbox{\scriptsize{start}}})$ and time $\TRx(T_{\mbox{\scriptsize{start}}})$, $B_0$ or more bits can be transmitted, i.e.,
\begin{small} 
\begin{equation}
T_{\mbox{\scriptsize{start}}}=\min\ t \ s.t.\  \TRx(t)g\Bigg{(} \dfrac{\ETx(t)}{\TRx(t)}\Bigg{)}\ge B_0.\label{online_T_start}
\vspace{-0.2cm}
\end{equation}
\end{small}
Starting at $T_{\mbox{\scriptsize{start}}}$, the algorithm transmits with power $l_1$, such that $\frac{\ETx(T_{\mbox{\scriptsize{start}}})}{l_1}g(l_1)=B_0$.
After $T_{start}$, at \textit{every} $\tau_j$, the transmission power is changed to $l_j$ such that
\vspace{-0.25cm} 
\begin{equation}
\frac{E_{\mbox{\scriptsize{rem}}}(\tau_j)}{l_j} g(l_j)= B_{\mbox{\scriptsize{rem}}}(\tau_j). 
\vspace{-0.2cm}
\end{equation} 
Transmission power is not changed at any time arrival at the receiver after $T_{\mbox{\scriptsize{start}}}$, because there is sufficient receiver time already available to finish transmission. 

\textit{Example:} Fig. \ref{figure_online_example} shows the output of the proposed online algorithm, \eqref{online_T_start} is not satisfied at time $\tau_0$, $r_1$, and $\tau_1$. At time $r_2$, \eqref{online_T_start} is satisfied and transmission starts with a power $l_1$ such that at rate $g(l_1)$, $B_0$ bits can be sent in $\ETx(r_2)/l_1$ time. Transmission power changes to $l_2$ at time $\tau_2$ such that $\frac{E_{\mbox{\scriptsize{rem}}}(\tau_2)}{l_2}g(l_2)=B_{\mbox{\scriptsize{rem}}}(\tau_2)$, and so on.
%
%
\begin{figure}
\centering
  	\centerline{\includegraphics[width=8cm]{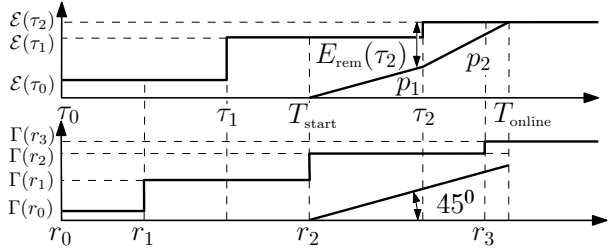}}
	\vspace{-0.25cm}
	\caption{An example for online algorithm.}
	\label{figure_online_example}
	\vspace{-0.6cm}
\end{figure}
%
%
%
%
%
\vspace{-0.1cm}
\begin{lemma}
The transmission power in the online algorithm is non-decreasing with time.
\label{online_power}
\vspace{-.2cm}
\end{lemma}
\begin{proof}
Combined with proof of Lemma 8.
\vspace{-0.25cm}
\end{proof}
\begin{lemma}
If power at time $t$ is $l$, then $\dfrac{\ETx(t)}{{l}}g(l) \le B_0,\;\;\forall\;\; t\in [T_{\mbox{\scriptsize{start}}}, T_{\mbox{\scriptsize{online}}}]$, with equality only at $t=T_{\mbox{\scriptsize{start}}}$.
\label{lemma_online_inequality}
\vspace{-0.1cm}
\end{lemma}
\begin{proof}
It is enough to prove that $\frac{g(l_i)}{l_i} \le \frac{B_0}{\ETx(b_i)}$ for $i\in\{1,..,M\} $, because both $l_i$ and $\ETx(t)$ remains constant in $t\in[b_i,b_{i+1})$. We prove this by induction on $i$ in $\{1,2..,M\}$. 

With $b_1=T_{\mbox{\scriptsize{start}}}$, the base case follows since at time $T_{start}$, $\frac{\ETx(T_{\mbox{\scriptsize{start}}})}{l_1}g(l_1)=B_0$. Now, assume $\frac{g(l_i)}{l_i}\le \frac{B_0}{\ETx(b_i)}$ to be true for $i=k-1$, $k\in \{2,..,M\}$. As $b_k=\tau_j$ for some $j$,
\begin{align*}
&{\frac{l_{k}}{g(l_{k})}=\frac{E_{\mbox{\scriptsize{rem}}}(b_{k})}{B_{\mbox{\scriptsize{rem}}}(b_{k})}=\frac{E_{\mbox{\scriptsize{rem}}}(b_{k-1})-l_{k-1} (b_k-b_{k-1})+E_{j}}{B_{\mbox{\scriptsize{rem}}}(b_{k-1})-g(l_{k-1}) (b_k-b_{k-1})},}
\\
&\stackrel{(a)}{=}\frac{l_{k-1}}{g(l_{k-1})}+\frac{E_{j}}{B_{\mbox{\scriptsize{rem}}}(b_{k-1})\gamma}
\stackrel{(b)}{>}\frac{\ETx(b_{k-1})}{B_0}+\frac{E_{j}}{B_0}=\frac{{\ETx(b_{k})}}{B_0}.
\end{align*}
where $(a)$ follows from $\frac{B_{\mbox{\scriptsize{rem}}}(b_{k-1})}{E_{\mbox{\scriptsize{rem}}}(b_{k-1})}=\frac{g(l_{k-1})}{l_{k-1}}$ and defining $\gamma=\left( 1-\frac{l_{k-1}(b_k-b_{k-1})}{E_{\mbox{\scriptsize{rem}}}(b_{k-1})} \right)< 1$, $(b)$ uses induction hypothesis along with $B_{\mbox{\scriptsize{rem}}}(b_{k-1})\gamma< B_0$. This completes the proof of Lemma \ref{lemma_online_inequality}. From equality $(a)$ we can see that $g(l_k)/l_k<g(l_{k-1})/l_{k-1}$. Hence, by monotonicity of $g(p)/p$, $l_{k}>l_{k-1}$. This proves Lemma \ref{online_power} as well.
\end{proof}
\begin{lemma} 
\vspace{-0.1cm}
For the online algorithm, $T_{\mbox{\scriptsize{start}}} <T_{\mbox{\scriptsize{off}}}$.
\label{online_time}
\vspace{-0.2cm}
\end{lemma}
\begin{proof}We use Contradiction. Suppose $T_{\mbox{\scriptsize{start}}} \ge T_{\mbox{\scriptsize{off}}}$. From \eqref{online_T_start}, either $T_{\mbox{\scriptsize{start}}}=\tau_i$ for some $i$ and/or $T_{\mbox{\scriptsize{start}}}=r_j$ for some $j$.
Let $T_{\mbox{\scriptsize{start}}}=\tau_i$. Since, the offline algorithm $\{\bm{p},\bm{s},N\}$ finishes before $T_{\mbox{\scriptsize{start}}}$, 
the maximum (cumulative) energy utilized by the optimal offline algorithm is at most the energy arrived till time $T_{\mbox{\scriptsize{start}}}^-$. So, $\sum_{i:p_i\neq 0}p_i(s_{i+1}-s_{i})\le \ETx(T_{\mbox{\scriptsize{start}}}^-)=\ETx(T_{\mbox{\scriptsize{start}}})-\ETx_i\neq \ETx(T_{\mbox{\scriptsize{start}}})$.
Similarly, if $T_{\mbox{\scriptsize{start}}}=r_j$, then the maximum time for which the receiver can be \textit{on} is $\TRx(T_{\mbox{\scriptsize{start}}}^-)$. So, $\sum_{i:p_i\neq 0}(s_{i+1}-s_{i})\le\TRx(T_{\mbox{\scriptsize{start}}}^-)=\TRx(T_{\mbox{\scriptsize{start}}})-\TRx_j\neq \TRx(T_{\mbox{\scriptsize{start}}})$.

Therefore, the total bits transmitted by the optimal offline algorithm $\{\bm{p},\bm{s},N\}$ is $\sum_{i=1,\ p_i\neq 0}^{N} g(p_i)(s_{i+1}-s_{i})$
\vspace{-0.2cm}
\begin{align}
&\nonumber \stackrel{(a)}{\le}g\left(\frac{\sum_{i:p_i\neq 0}p_i(s_{i+1}-s_{i})}{\sum_{j:p_j\neq 0}(s_{j+1}-s_{j})}\right)\sum_{j:p_j\neq 0} (s_{j+1}-s_{j}),
\\
&\stackrel{(b)}\le g\left(\frac{\ETx(T_{\mbox{\scriptsize{start}}}^-)}{\TRx(T_{\mbox{\scriptsize{start}}}^-)}\right)\TRx(T_{\mbox{\scriptsize{start}}}^-)\stackrel{(c)}{<}B_0,\label{online_eq_2}
\end{align}
where $(a)$ follows from Jensen's inequality since $g(p)$ is concave, $(b)$ follows from monotonicity of $g(p)/p$,
and $(c)$ follows from \eqref{online_T_start}. \eqref{online_eq_2} says that offline policy transmits less than $B_0$ bits and therefore, we arrive at a contradiction.
\end{proof}
\begin{theorem}\label{thm:onlinecomp}
\vspace{-0.1cm}
The proposed online algorithm is $2$-competitive.
\end{theorem}
\begin{proof}



\vspace{-0.2cm}
Let the online algorithm transmit with power $l_k$ at time $T_{\mbox{\scriptsize{off}}}^-$.
Since $T_{\mbox{\scriptsize{start}}}<T_{\mbox{\scriptsize{off}}}$ by Lemma \ref{online_time}, ${l}_k > 0$. Let $b_k<T_{\mbox{\scriptsize{off}}}$ be the time where transmission starts with power $l_k$. 
By definition, $\sum_{i=k}^{M}g(l_i)(b_{i+1}-b_i)=B_{\mbox{\scriptsize{rem}}}(b_k)$.
From  Lemma \ref{online_power},
\begin{equation}
(b_{N+1}-b_k)\le\frac{B_{\mbox{\scriptsize{rem}}}(b_k)}{g(l_k)}=
\frac{E_{\mbox{\scriptsize{rem}}}(b_k)}{l_k}\le \frac{\ETx(b_k)}{l_k}\le \frac{\ETx(T_{\mbox{\scriptsize{off}}}^-)}{l_k}.
\label{eq_online_time_1}  
\vspace{-0.3cm}
\end{equation}
Applying Lemma \ref{lemma_online_inequality} at time $T_{\mbox{\scriptsize{off}}}^-$,
\vspace{-0.1cm}
\begin{equation}
\frac{\ETx(T_{\mbox{\scriptsize{off}}}^-)}{l_k}g(l_k)\le B_0\stackrel{(a)}{\le}T_{\mbox{\scriptsize{off}}}\; g\left(\frac{\ETx(T_{\mbox{\scriptsize{off}}}^-)}{T_{\mbox{\scriptsize{off}}}}\right),
\label{eq_online_time_2}
\vspace{-0.2cm}
\end{equation}
where $(a)$ holds because the maximum number of bits sent by the optimal offline policy by time $T_{\mbox{\scriptsize{off}}}$ can be bounded by $T_{\mbox{\scriptsize{off}}}\, g\left(\frac{\ETx(T_{\mbox{\scriptsize{off}}}^-)}{T_{\mbox{\scriptsize{off}}}}\right)$ due to concavity of $g(p)$. \vspace{-0.2cm}
By monotonicity of $g(p)/p$, from \eqref{eq_online_time_2}, it follows that  
$\frac{\ETx\left(T_{\mbox{\scriptsize{off}}}^-\right)}{l_k}\le T_{\mbox{\scriptsize{off}}}$.
Combining this with \eqref{eq_online_time_1},
$(b_{{N}+1}-b_k)\le T_{\mbox{\scriptsize{off}}}$.
As $b_k<T_{\mbox{\scriptsize{off}}}$, we calculate the competitive ratio as,
\vspace{-0.1cm}
\begin{equation*}
r=\displaystyle\max_{\ETx(t),\TRx(t)\hspace{0.5mm} \forall t}\dfrac{T_{\mbox{\scriptsize{online}}}}{T_{\mbox{\scriptsize{off}}}} = \dfrac{(b_{{N}+1}-b_k)+b_k}{T_{\mbox{\scriptsize{off}}}} <  2.
\vspace{-0.5cm}
\end{equation*}
\end{proof}
{\it Discussion:} Theorem \ref{thm:onlinecomp} is a significant result, since it tells us that the proposed online (causal) algorithm will finish in at most twice the time an optimal offline algorithm takes knowing all energy arrivals non-causally. Moreover, the online algorithm is independent of the energy arrival distributions both at the transmitter and the receiver, so has built-in robustness. Also, note that the proof of Theorem \ref{thm:onlinecomp} does not explicitly require to know the exact structure of the optimal offline algorithm.

\bibliographystyle{IEEEtran}
\bibliography{IEEEabrv,../Work/TIFR/Research/Research}
 
\end{document}